\documentclass[runningheads]{llncs}

\usepackage{times}
\usepackage[latin1]{inputenc}
\usepackage[english]{babel}
\usepackage{amsmath}
\usepackage{colortbl}
\usepackage{txfonts}
\usepackage{xypic}
\usepackage{multirow}
\usepackage{textcomp}
\usepackage{calc}
\usepackage{stmaryrd}
\usepackage{marvosym}
\usepackage{pgf}
\usepackage[mathcal]{euscript}

\usepackage{ifthen}
\usepackage{mdwlist}
\usepackage{enumerate}
\usepackage{paralist}

\usepackage{tikz}
\usetikzlibrary{arrows,automata}
\usetikzlibrary{decorations.pathmorphing}
\usetikzlibrary{decorations.markings}
\usetikzlibrary{decorations.shapes}

\DeclareMathAlphabet{\mathpzc}{OT1}{pzc}{m}{it}

\newcommand{\mrk}[1]{{#1}'}
\newcommand{\oldstack}[3]{%
{\ifthenelse{\equal{#1}{1}}{%
\mrk{#2}
}%
{#2}}_{#3}%
}
\newcommand{\stack}[3]{%
[%
{\ifthenelse{\equal{#1}{1}}{%
\mrk{#2}
}%
{#2}}\ {#3}%
]%
}

\newcommand{\va}[1]{\stackrel{#1}{\longrightarrow}}
\newcommand{\ourpath}[1]{\stackrel{#1}{\leadsto}}
\newcommand{\flush}[1]{\stackrel{#1}{\Longrightarrow}}

\newcommand{\nont}[4]{\langle^{#1} #2, #3 {}^{#4} \rangle}
\newcommand{\chain}[3]{\langle^{#1} #2 {}^{#3} \rangle}

\newcommand{\config}[2]{\langle #1\ , \ #2 \rangle}
\newcommand{\symb}[1]{\mathop{symbol}(#1)}
\newcommand{\state}[1]{\mathop{state}(#1)}

\newcommand{\tp}[1]{\mathop{top}(#1)}
\newcommand{\inpt}[1]{\mathop{input}(#1)}

\newcommand{\stpair}[2]{( {#1},{#2} )} 
\newcommand{\pair}[2]{\langle {#1},{#2}\rangle} 

\newcommand{\der}[1]{ \stackrel {{#1}} \Rightarrow}
\newcommand{\derh}[1]{ \stackrel {[{#1}]} \Rightarrow}
\newcommand{\comp}[1]{ \stackrel {{#1}} \vdash } 
\newcommand{\comph}[1]{ \stackrel {[{#1}]} \vdash } 

\tikzset{
	path/.style={dotted},
	every edge/.style={draw,solid},
	normal/.style={solid},
}

\begin{document}

\title{Precedence Automata and Languages\thanks{This is an extended version
of the paper which appeared in Proceedings of CSR 2011, 6th International Computer Science Symposium in Russia,
Lecture Notes in Computer Science, vol. 6651, pp. 291-304, 2011. 
In particular, Theorem~\ref{th:nondet} has been corrected and a complete proof is given in Appendix.}}
\author{Violetta Lonati\inst{1}, Dino Mandrioli\inst{2}, Matteo
Pradella\inst{2}}

\institute{
  DSI - Universit\`a degli Studi di Milano,
  via Comelico 39/41, Milano, Italy\\
  \email{lonati@dsi.unimi.it}
\and
  DEI - Politecnico di Milano,
  via Ponzio 34/5, Milano, Italy\\
  \email{\{dino.mandrioli, matteo.pradella\}@polimi.it}
}

\maketitle

\begin{abstract}
Operator precedence grammars define a classical Boolean and deterministic
context-free family (called Floyd languages or FLs). FLs have been shown to
strictly include the well-known visibly pushdown languages, and enjoy the same
nice closure properties. We introduce here Floyd automata, an equivalent
operational formalism for defining FLs. This also permits to extend the class to
deal with infinite strings to perform for instance model checking.\\
\\
{\bf Keywords: }
Operator precedence languages, Deterministic Context-Free languages,
Omega languages, Pushdown automata.
\end{abstract}

\section{Introduction}
\label{sec:intro}

The history of formal language theory has always paired two main and
complementary formalisms to define and process --not only formal-- languages:
grammars or syntaxes and abstract machines or automata. The power and the
complementary benefits of these two formalisms are so evident and well-known that
it is certainly superfluous to remind them here. Also universally known are the
conceptual relevance and practical impact of the family of context-free
languages and the corresponding grammars paired with pushdown automata.

Among the many subfamilies that have been introduced throughout the last decades
with various goals, operator precedence grammars, herewith renamed Floyd grammars
(FGs) in honor of their inventor~\cite{Floyd1963}, represent a pioneering model mainly aimed
at deterministic --and therefore efficient-- parsing. Visibly pushdown languages
(VPLs) are a much more recent subfamily of (deterministic) context-free
languages introduced in the seminal paper~\cite{AluMad04} with the goal of extending the
typical closure properties of regular languages to larger families of languages
accepted by infinite-state machines; a major practical result is the possibility
of extending such powerful verification technique as model checking beyond the
scope of finite state machines. Along the usual tradition, VPLs have been
characterized both in terms of abstract machines, the visibly pushdown automata
(VPAs), and by means of a suitable subclass of context-free grammars.

Rather surprisingly, instead, investigation of the basic --and nice, indeed--
properties of FGs has been suspended, probably as a consequence of the advent of
other, more general, parsing techniques, such as LR parsing~\cite{GruneJacobs:08}.
Although FGs generate obviously a subclass of deterministic CF languages and
therefore can be parsed by any deterministic pushdown machine, typically a
shift-reduce one~\cite{GruneJacobs:08}, we are not aware of a family of automata that perfectly
matches the generative power of this class of grammars.
On the other hand, operator precedence parsers are
still used today, thanks to their elegant simplicity and efficiency. For
instance,
they are present in Parrot, Perl 6's virtual machine, as part of the Parser Grammar
Engine (PGE); in GCC's C and C++ hand-coded parsers, for managing arithmetic
expressions.\footnote{The interested reader may find more information at {\em
http://gcc.gnu.org}, and {\em http://www.parrot.org}, respectively.}

Quite recently we realized strong relations between these two seemingly
unrelated families of languages; precisely we showed that: VPLs are a proper
subclass of languages defined by FGs (i.e. Floyd Languages, or FLs in short), and coincide with those languages that can be generated by FGs
characterized by a well precise shape of operator precedence matrix (OPM). The
inclusion relation is effective in that a FG can be algorithmically derived form
a VPA and conversely a VPA can be obtained by a FG whose OPM satisfies the
restriction \cite{CrespiMandrioliWORDS2009}.

FLs enjoy all typical closure properties of regular languages that motivated the
study of VPLs and other related families \cite{Berstel:2001:BGT,conf/mfcs/NowotkaS07,caucal:DSP:2008:1743}. Precisely, closure w.r.t.
Boolean operations was proved a long time ago in \cite{Crespi-ReghizziMM1978}, whereas closure under
concatenation, Kleene star, and other typical algebraic operations has been
investigated only recently under the novel interest ignited by the above remark
\cite{Crespi-ReghizziM10}. Thus, the old-fashioned FLs turned out to be the
largest known class of
deterministic context-free languages that enjoy closure under all traditional
language operations.  Another reason why, in our opinion, FLs are far from
obsolete and uninteresting in these days is that, unlike most other
deterministic languages of practical use, they can be parsed not necessarily
left-to-right, thus offering interesting opportunities, e.g., to exploit
parallelism and incrementality \cite{GruneJacobs:08}.

In this paper we provide another missing tile of the ``old and new puzzle'',
namely we introduce a novel class of stack-based automata perfectly carved on
the generation mechanism of FGs, which too we name in honor of Robert Floyd. Not
surprisingly they inherit some features of VPAs (mainly a clear separation
between push and pop operations) and maintain some typical behavior of
shift-reduce parsing algorithms; however, they also exhibit some distinguishing
features and imply some non-trivial technicalities to derive them automatically
from FGs and conversely.

The availability of a precise family of automata allows to apply to FLs the now
familiar $\omega$-extension --a further extension of Kleene $*$ operation--,
i.e., the definition of languages of infinite strings and the various criteria
for their acceptance or rejection by recognizing devices. $\omega$-languages are
now more and more important to deal with never-ending computations such as
operating systems, web-services, embedded applications, etc. Thus, we also
introduce the $\omega$-version of FLs and we show their potential in terms of
modeling the behavior of some realistic systems.

The paper is structured as follows: Section~\ref{sec:prelim} recalls basic
definitions on Floyd's grammars;
Section~\ref{sec:aut} introduces Floyd
automata (FAs) and shows that, as well as FSMs and VPAs, but unlike pushdown automata,
their deterministic version is not less powerful than the nondeterministic
counterpart; Section~\ref{sec:aut_gr} provides effective constructions to derive
a FA from a FG and conversely; Section~\ref{sec:omega} extends the definition of
FLs to sets of infinite strings by applying to FAs the well-known concepts of
$\omega$-behavior and acceptance; finally Section~\ref{sec:conclusions} draws
some conclusions.

\section{Preliminaries}
\label{sec:prelim}



Let $\Sigma$ be an alphabet. The empty string is denoted $\varepsilon$.
A \textit{context-free} (CF) grammar is  a 4-tuple $G=(N, \Sigma, P, S)$, where $N$ is the nonterminal alphabet, $P$ the rule (or production) set, and $S$ the axiom. 
An \textit{empty rule} has $\varepsilon$ as the right hand side (r.h.s.). 
A \textit{renaming rule} has one nonterminal as r.h.s. A grammar is \textit{reduced} if every rule can be used to generate some string in $\Sigma^\ast$. It is \textit{invertible} if no two rules have identical r.h.s.

The following naming convention will be adopted, unless otherwise specified: lowercase Latin letters $a,b,\ldots$ denote  terminal characters; uppercase Latin letters $A,B, \ldots$ denote nonterminal characters; letters $u,v,\ldots$ denote terminal strings; and Greek letters $\alpha, \ldots, \omega$ denote strings over $\Sigma \cup N$. The strings may be empty, unless stated otherwise.


A rule is in \textit{operator form} if its r.h.s has no adjacent nonterminals; an \textit{operator grammar} (OG) contains just such rules. Any CF grammar admits an equivalent OG, which can be also assumed to be invertible \cite{Harrison78,Salomaa73}.




The  coming definitions for operator precedence grammars \cite{Floyd1963}, here renamed \textit{Floyd Grammars} (FG), 
are from \cite{Crespi-ReghizziMM1978}. 
We refer the reader unfamiliar with precedence grammars and parsing techniques
to
\cite{GruneJacobs:08}, that contains an easily readable, practical description
of FGs.

For an OG $G$ and a nonterminal $A$, the \textit{left and right terminal sets} are
\[
  \mathcal{L}_G(A)  = \{a\in\Sigma \mid A \stackrel \ast \Rightarrow B a\alpha\} \qquad
  \mathcal{R}_G(A) =  \{a\in\Sigma \mid A \stackrel \ast \Rightarrow \alpha a B\}
\]
where $B\in N\cup\{\varepsilon\}$ and $\Rightarrow$ denotes the derivation
relation. 
%
%
The grammar name $G$ will be omitted unless necessary to prevent confusion.

R. Floyd took inspiration from the traditional notion of precedence between arithmetic operators in order to define a broad class of languages, such that the shape of the derivation tree is solely determined by a binary relation between terminals that are consecutive, or become consecutive after a bottom-up reduction step.

For an OG $G$, let $\alpha, \beta$ range over  $(N \cup \Sigma)^{\ast}$ and $a,b\in \Sigma$. Three binary operator precedence (OP) relations are defined: \label{PrecRelat}
\begin{eqnarray*}
 \nonumber \text{equal in precedence: }  & a\doteq b  \iff&   \exists A\to\alpha aBb\beta, B\in N\cup\{\varepsilon\} \\
   \text{takes precedence: } & a\gtrdot b  \iff  &\exists A\to\alpha Db\beta, D\in N \text{ and } a\in \mathcal{R}_G(D) \\
 \nonumber    \text{yields precedence: }& a\lessdot b  \iff & \exists A\to\alpha aD\beta, D\in N \text{ and } b\in \mathcal{L}_G(D)
\end{eqnarray*}
For an OG $G$, the \textit{operator precedence matrix} (OPM) $M=OPM(G)$ is a $|\Sigma| \times |\Sigma|$ array that with each ordered pair $(a,b)$ associates the set $M_{ab}$ of OP relations holding between $a$ and $b$. 

\begin{definition}\label{defFloydGr}
$G$ is an \emph{operator precedence} or \emph{Floyd grammar} (FG) if, and only if, $M=OPM(G)$ is a \textit{conflict-free} matrix, i.e., $\forall a,b$, $|M_{ab}|\leq 1$. 
\end{definition}

\begin{example}
\label{ex:expr}
Arithmetic expressions with prioritized operators, a classical construct, are
presented in a simple variant without parentheses.
Figure~\ref{fig:exp} presents the productions of the grammar (left) and the derivation tree of expression 
$n + n \times n$ (center). We see that $\times \ \dot= \ n$ because they
appear in the right-hand side of the same production. Analogously, $+ \lessdot
n$ since
$+$ is sibling of a node with label $T$ and $n \in  \mathcal{L}_G(T)$.
The complete OPM is shown in Figure~\ref{fig:exp} (right).
\end{example}

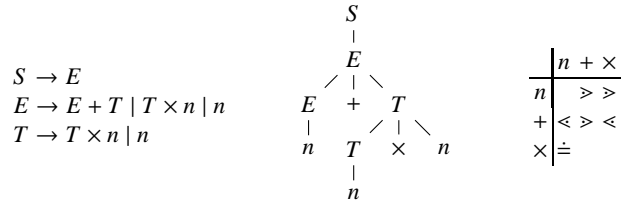
\begin{figure}
\begin{center}
\begin{tabular}{m{0.3\textwidth}m{0.25\textwidth}m{0.2\textwidth}}
\quad 
$\begin{array}{l}
S\to E\\
E\to E + T \mid T \times n \mid n\\ 
T\to T \times n \mid n 
\end{array}$
&
\quad
\begin{tikzpicture}[scale=0.4]
\node {$S$}
child {
        node {$E$}
	child {
		node {$E$}
		child{ 
			node{$n$}
		}
	}
	child {
		node {$+$}
	}
	child {
		node {$T$}
		child {
			node {$T$}
			child {
				node {$n$}
		                edge from parent [normal]
			}
		}
		child{
			node {$\times$}
		}
		child{
			node {$n$}
		}
	}
};
\end{tikzpicture}
&
\quad 
$
\begin{array}{c|ccc}
    & n & + & \times  \\
\hline
  n &   & \gtrdot & \gtrdot     \\
  + &  \lessdot &\gtrdot  & \lessdot     \\
  \times& \dot=  &   &       \\
\end{array}
$
\end{tabular}
\end{center}
\vspace{-.5cm}\caption{The Floyd grammar for arithmetic expressions without parentheses.}
\label{fig:exp}
\end{figure}

\medskip

The equal in precedence relations of a FG alphabet  are connected with an
important parameter of the grammar, namely the length of the right hand sides of
the rules. Clearly, a rule $A \to A_1 a_1 \ldots A_t a_t A_{t+1}$, where each
$A_i$ is a possibly missing nonterminal, is associated with relations $a_1 \dot=
a_2\dot= \ldots \dot= a_t$. If the $\dot=$ relation is cyclic, there is no
finite bound on the length of the r.h.s of a production. Otherwise the length is
bounded by $2\cdot c+1$, where $c\geq 1$ is the length of the longest $\dot=$-chain. 
In this paper, for the sake of simplicity and brevity we assume that all precedence matrices are $\doteq$-cycle free. 
In the case of FGs this prevents the risk of r.h.s of unbounded
length~\cite{Crespi-ReghizziMM1978}, 
in the case of FAs we will see that it avoids a priori the risk of an unbounded sequence of push operations onto the stack matched by only one pop operation.
The hypothesis of $\doteq$-cycle freedom could be replaced by weaker ones,
such as a bound on r.h.s, as it happens with FGs, at the price of heavier notation, constructions, and proofs.  

%


\begin{definition}
\label{defNormalForms}
A FG is in \textit{Fischer normal form} \cite{Fischer69} if it is invertible, the axiom $S$ does not occur in the r.h.s. of any rule, no empty rule exists  except possibly $S\to \varepsilon$, the other rules having  $S$ as l.h.s are renaming, and no other renaming rules exist.

\end{definition}
%

OPMs play a fundamental role in deterministic parsing of FGs. Thus in the view
of defining automata to parse FLs we pair them with the alphabet somewhat
mimicking VPL's approach where the terminal alphabet is partitioned into
calls, returns, and internals \cite{jacm/AlurM09}.
To this goal, we use a special symbol \# not in $\Sigma$ to 
mark the beginning and the end of any string. This is consistent with the
typical operator parsing technique that requires the lookback and lookahead of
one character to determine the precedence relation \cite{GruneJacobs:08}.
The precedence relation in the OPM are extended to include \# in the
normal way.

\begin{definition}
An \emph{operator precedence alphabet} is a pair $(\Sigma, M)$ where $\Sigma$ is an alphabet and 
$M$ is a conflict-free \textit{operator precedence matrix},
i.e. a $|\Sigma \cup \{ \# \}|^2$ array that with each ordered pair $(a,b)$ associates at most 
one of the operator precedence relations: $\doteq$, $\lessdot$ or $\gtrdot$. 
\end{definition}
\noindent For $u,v \in\Sigma^*$ we write $u \lessdot v$ if $u = xa$ and $v = by$ with $a \lessdot b$. Similarly for the other precedence relations.

\bigskip

\section{Floyd automata}
\label{sec:aut}

\begin{definition}
A nondeterministic \emph{precedence automaton} (or Floyd automaton) is given by a tuple:
$\mathcal A = \langle \Sigma, M, Q, I, F, \delta \rangle $ where:
\begin{itemize}
\item $(\Sigma, M)$ is a precedence alphabet,
\item $Q$ is a set of states (disjoint from $\Sigma$),
\item $I \subseteq Q$ is a set of initial states,
\item $F \subseteq Q$ is a set of final states,
\item $\delta : Q \times ( \Sigma \cup Q) \rightarrow 2^Q$ is the transition function.
\end{itemize}
\end{definition}

\noindent The transition function can be seen as the union of two disjoint functions:
\[
\delta_{\text{push}}: Q \times \Sigma \rightarrow 2^Q
\qquad 
\delta_{\text{flush}}: Q \times Q \rightarrow 2^Q
\]
A nondeterministic precedence automaton can be represented by a graph with $Q$ as the set of vertices and
$\Sigma \cup Q$ as the set of edge labellings: 
there is an edge from state $q$ to state $p$ labelled by $a \in \Sigma$ if and only if $p \in \delta_{push}(q,a)$ and 
there is an edge from state $q$ to state $p$ labelled by $r \in Q$ if and only if $p \in \delta_{flush}(q,r)$.
To distinguish flush transitions from push transitions we denote the former ones by a double arrow.

To define the semantics of the automaton, we introduce some notations.
We use 
letters $p, q, p_i, q_i, \dots $ for states in $Q$ and 
we set $\mrk{\Sigma} = \{\mrk a \mid a \in \Sigma \}$; symbols in $\Sigma'$ are
called {\em marked} symbols.
%
Let $\Gamma = (\Sigma \cup \mrk \Sigma \cup \{\#\}) \times Q$; we denote symbols in $\Gamma$ as $\stack 0aq$,
$\stack 1aq$, or $\stack 0\#q$, respectively.
We set $\symb {\stack 0aq} = \symb {\stack 1aq} = a$, $\symb {\stack 0\#q}=\#$, and
$\state {\stack 0aq} = \state {\stack 1aq} = \state {\stack 0\#q} = q$.
%
%
Given a string $\beta = B_1 B_2 \dots B_n$ with $B_i \in \Gamma$, we set 
$\state \beta = \state{B_n}$.

We call a \emph{configuration} any pair $C = \config {\beta} {w}$,
where $\beta = B_1 B_2 \dots B_n\in \Gamma^*$, $\symb{B_1} = \#$, and $w = a_1 a_2 \dots a_m \in \Sigma^*\#$.
A configuration represents both the contents $\beta$ of the stack and the part of input $w$ still to process.
We also set
$\tp C = \symb{B_n}$ and $\inpt C = a_1$.

A computation of the automaton is a finite sequence of moves $C \vdash C_1$; there are three kinds of moves, depending on the precedence relation between $\tp C$ and $\inpt C$:

\smallskip

\noindent {\bf push move:} 

if $\tp C \doteq \inpt C$ then 
$
\config{\beta}{aw} \vdash \config{\beta \stack 0aq }{w}, \
\forall q \in \delta_{push}(\state\beta,a)$;
\smallskip

\noindent {\bf mark move:} 

if $\tp C \lessdot \inpt C$ then
$
\config{\beta}{aw} \vdash \config{\beta \stack 1aq }{w}, \
\forall q \in \delta_{push}(\state\beta,a)$;

\smallskip

\noindent {\bf flush move:} 

if $\tp C \gtrdot \inpt C$ then
let $\beta = B_1 B_2 \dots B_n$ with ${B_j} = \stack 0{x_j}{q_j}$, $x_j \in \Sigma \cup \mrk\Sigma$ 
and let $i$ the greatest index such that
$B_i$ belongs to $\mrk\Sigma \times Q$. Then
\[
\config{\beta}{aw} \vdash \config{B_1 B_2 \dots B_{i-2}\stack
0{x_{i-1}} }{aw}, \
\forall q \in \delta_{flush}(q_n,q_{i-1}).
\]

\smallskip

Push and mark moves both push the input symbol on the top of the stack, together with the new state computed by $\delta_{push}$; such moves differ only in the marking of the symbol on top of the stack.
The flush move is more complex: the symbols on the top of the stack are removed until the first marked symbol (\emph{included}),
and the state of the next symbol below them in the stack is updated by $\delta_{flush}$ according to the pair of states that delimit the portion of the stack to be removed; notice that in this move the input symbol is not relevant and it remains available for the following move.

Finally, we say that a configuration $\stack 0{\#}{q_I}$ is {\em starting} if $q_I \in I$ and
a  configuration $\stack 0{\#}{q_F}$ is {\em accepting} if $q_F \in F$.
The language accepted by the automaton is defined as:
\[
L(\mathcal A) = \left\{ x \mid  \config {\stack 0\#{q_I}} {x\#}  \comp * 
\config {\stack 0\#{q_F}} \# , q_I \in I, q_F \in F \right\}.
\]

\begin{example}
The automaton depicted in Figure \ref{ex:primo} accepts the Dyck
language $L_D$ of balanced strings of parentheses, with two parentheses pairs $a, \underline{a}$, and $b, \underline{b}$.
The same figure also shows an accepting computation on input
$aba\underline{a}\underline{b}\underline{a}a\underline{a}$.
\end{example}


\begin{figure}
\begin{center}

\begin{tabular}{l}
\begin{tikzpicture}[every edge/.style={draw,solid}, node distance=4cm, auto, 
                    every state/.style={draw=black!100,scale=0.5}, >=stealth]

\node[initial by arrow, initial text=,state,accepting] (S) {{\huge $q_0$}};
\node[state] (E) [right of=S, xshift=0cm] {{\huge $q_1$}};

\path[->]
(S) edge [bend left]  node {$a, b$} (E)
(E) edge [loop below] node {$a, \underline{a}, b, \underline{b}$} (E)
(E) edge [loop right, double] node {$ q_1$} (E)
(E) edge [double, below]  node {$q_0$} (S) ;
\end{tikzpicture}
\\
$
\begin{array}{c|ccccc}
      & a & \underline{a} & b & \underline{b} & \# \\
\hline
a & \lessdot & \dot= & \lessdot &  &  \\
\underline{a}  & \lessdot & \gtrdot & \lessdot &\gtrdot & \gtrdot \\
b & \lessdot & & \lessdot & \dot=  & \\
\underline{b}  & \lessdot & \gtrdot & \lessdot & \gtrdot & \gtrdot \\
\#       & \lessdot & & \lessdot & & \dot=  \\
\end{array}
$
\end{tabular}
\qquad\quad
$
\begin{array}{llcr}
 & \langle \stack 0\# {q_0}     & , &  
  aba\underline{a}\underline{b}\underline{a}a\underline{a}\# \rangle \\ 
\text{mark} & \langle \stack 0\# {q_0} \stack 1 a {q_1}     & , &      ba\underline{a}\underline{b}\underline{a}a\underline{a}\# \rangle \\
\text{mark} & \langle \stack 0\# {q_0} \stack 1 a {q_1} \stack 1 b {q_1}     & , &      a\underline{a}\underline{b}\underline{a}a\underline{a}\# \rangle \\
\text{mark} & \langle \stack 0\# {q_0} \stack 1 a {q_1} \stack 1 b {q_1} \stack 1 a {q_1}     & , &      \underline{a}\underline{b}\underline{a}a\underline{a}\# \rangle \\
\text{push} & \langle \stack 0\# {q_0} \stack 1 a {q_1} \stack 1 b {q_1} \stack 1 a {q_1} 
  \stack 0 {\underline{a}} {q_1}     & , &      \underline{b}\underline{a}a\underline{a}\# \rangle \\
\text{flush} & \langle \stack 0\# {q_0} \stack 1 a {q_1} \stack 1 b {q_1}     & , &      \underline{b}\underline{a}a\underline{a}\# \rangle \\
\text{push} & \langle \stack 0\# {q_0} \stack 1 a {q_1} \stack 1 b {q_1} \stack 0 {\underline{b}} {q_1}     & , &      \underline{a}a\underline{a}\# \rangle \\
\text{flush} & \langle \stack 0\# {q_0} \stack 1 a {q_1}     & , &      \underline{a}a\underline{a}\# \rangle \\
\text{push} & \langle \stack 0\# {q_0} \stack 1 a {q_1} \stack 0 {\underline{a}} {q_1}     & , &      a\underline{a}\# \rangle \\
\text{mark} & \langle \stack 0\# {q_0} \stack 1 a {q_1} \stack 0 {\underline{a}} {q_1} \stack 1 a {q_1}     & , &      \underline{a}\# \rangle \\
\text{push} & \langle \stack 0\# {q_0} \stack 1 a {q_1} \stack 0 {\underline{a}} {q_1} \stack 1 a {q_1} \stack 0 {\underline{a}} {q_1}     & , &      \# \rangle \\
\text{flush} & \langle \stack 0\# {q_0} \stack 1 a {q_1} \stack 0 {\underline{a}} {q_1}     & , &      \# \rangle \\
\text{flush} & \langle \stack 0\# {q_0}     & , &  \# \rangle \\
\end{array}
$

\caption{Automaton, precedence matrix, and example of computation for language $L_D$.}\label{ex:primo}
\end{center}
\end{figure}
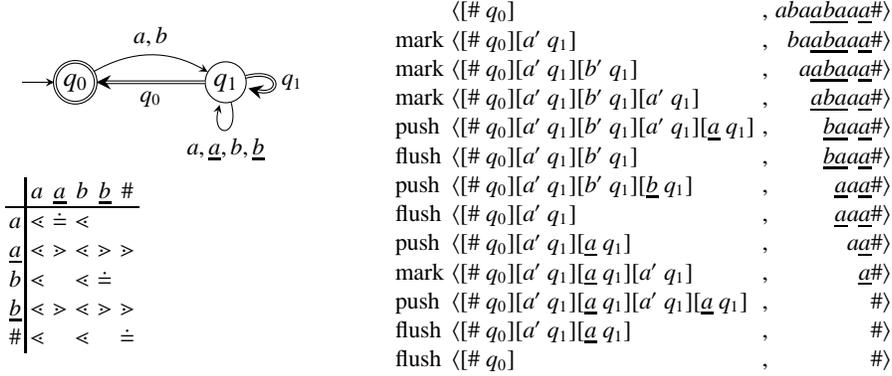



A Floyd automaton is called \emph{deterministic} when $\delta_{\text{push}}(q,a)$ and $\delta_{\text{flush}}(q,p)$
have at most one element, for every $q,p \in Q$ and $a \in \Sigma$, and $I$ is a singleton.
Here we prove that deterministic Floyd automata are equivalent to nondeterministic ones,
with a power-set construction similar to the one used for classical finite state automata.

\begin{theorem}
\label{th:nondet}
Deterministic Floyd automata are equivalent to nondeterministic ones.
\end{theorem}

\noindent
Given a nondeterministic automaton $\mathcal A = \langle \Sigma, M, Q, I, F, \delta \rangle $,
consider the deterministic automaton $\tilde{\mathcal A} = \langle \Sigma, M, \tilde Q, \tilde I, \tilde F, \tilde \delta \rangle $ where:

\begin{itemize}
	\item $\tilde Q = \hat{\Sigma} \times 2^{Q \times (Q \cup \{\bot\} )}$, 
	where $\hat{\Sigma} = (\Sigma \cup \left\{ \# \right\} )$, 
	$Q \cap \{  \bot  \} = \emptyset$, 
	and $\bot$ is a symbol that stands for the baseline of the computations (i.e. the pseudo-state before the initial states),

	\item $\tilde I = \langle \# ,  I \times \{  \bot \} \rangle $ is the initial state of $\mathcal {\tilde{A}}$,

	\item $\tilde F$ is the set of pairs $\pair{\#}{K}$ such that there exists $q \in F$ with $\stpair{q}{\bot} \in K.$
	
	\item $\tilde\delta : \tilde Q \times ( \Sigma \cup \tilde Q) \rightarrow \tilde Q$ is the transition function defined as follows.
The push transition ${\tilde\delta}_{\text{push}}: \tilde Q \times \Sigma \rightarrow \tilde Q$ is defined by
\[
{\tilde \delta}_{\text{push}}  (\pair{b}{K}, a) = 
\left\langle a, 
\bigcup_{ \stpair{q}{p} \in K} \left\{ 
\stpair{h}{t} \mid h \in \delta_{\text{push}}(q,a) \text{ and } t = \left[  
\begin{array}{ll}
q & \text{if } b \lessdot a\\
p & \text{if } b \doteq a
\end{array}
\right. \\
\right\}
\right\rangle
\]
The flush transition ${\tilde\delta}_{\text{flush}}: \tilde Q \times \tilde Q \rightarrow \tilde Q$ is defined as follows:
\[
{\tilde{\delta}}_{\text{flush}} (\pair{b}{K_1}, \pair{a}{K_2}) =
\left\langle a, 
\bigcup_{ \stpair rq \in K_1, \stpair qp \in K_2} \left\{ \stpair{h}{p} \mid h \in \delta_{\text{flush}} (r,q) \right\}
\right\rangle.
\]

\end{itemize}

\noindent
The proof of the equivalence between $\mathcal A$ and $\tilde{\mathcal A}$ is given in Appendix.

\section{Floyd automata vs Floyd grammars}
\label{sec:aut_gr}

The main result of this paper is the perfect match between FGs and FAs.

\subsection{From Floyd grammars to Floyd automata}
\label{sec:gr2aut}

\begin{theorem}
\label{GtoA}
Any $L$ generated by a Floyd grammar can be recognized by a Floyd automaton
\end{theorem}

We provide a constructive proof of the theorem: given a Floyd grammar $G$ we build
an equivalent nondeterministic Floyd automaton $\mathcal A = \langle  \Sigma, M, Q, I, F, \delta \rangle$, 
whose precedence matrix $M$ is the same as the one associated with $G$. 
A successful computation of $\mathcal A$ will correspond to a derivation tree in $G$:
intuitively, a push transition tries to guess the parent of the symbol currently under the input head 
(i.e. it determines the l.h.s of a rule of $G$ whose r.h.s contains the current symbol);
a flush transition is performed whenever the r.h.s of a rule is completed, and determines the corresponding l.h.s., 
thus confirming some previous guesses.

In order to keep the construction as simple as possible, we avoid introducing any optimization.
Also, without loss of generality, we assume that the grammar $G = \langle \Sigma, N, P, S \rangle$ 
satisfies the following properties: the axiom $S$ does not occur in the r.h.s. of any rule, 
no empty rule exists except possibly $S\to \varepsilon$, the other rules having $S$ as l.h.s are renaming, and no other renaming rules exist
(in other words, we assume that the $G$ is in Fischer normal form except it is not necessarily invertible).

First of all, we introduce some notation. Enumerate the productions as follows: for any nonterminal $A \in N$,
let $P_1(A),$ $P_2(A),$ $\dots P_{n(A)}(A)$ be the productions having $A$ as l.h.s.
(i.e. $n(A)$ is the number of productions having $A$ as l.h.s.).
Then, consider the set of \emph{extended nonterminals} $EN = \{ A_i \mid A \in N, i = 1, 2, \dots n(A) \}$
and define $Q =  EN \times (EN \cup \{ \perp \}) $, where $\perp$ is a new symbol 
whose meaning is \emph{undefined}. To distinguish between nonterminals and extended nonterminals, 
we will use capital letters $A,B,C, \dots$ and $X,Y,Z, \dots$, respectively.

When considering derivation trees of $G$, we label internal nodes with extended nonterminals (where the subscript of the nonterminal corresponds to the rule applied in the node). 
Moreover, with a slight abuse of notation, we sometimes confuse nodes and their labels, using the above convention also for internal nodes and leaves.


To define the push transition function $\delta_{push} : Q \times \Sigma \to 2^Q$,
consider any derivation tree $\tau$ of $G$ with any leaf $a$ and let  $X$ be $a$'s  parent in $\tau$.
Figure~\ref{fig:tree_push} represents the various configurations that $\tau$ may exhibit.
\begin{itemize}
\item Case 0: if there is no leaf that precedes $a$ in the in-order visit of
$\tau$ and has depth not greater than $a$'s depth,
then let $Y$ be the topmost ancestor of $X$, i.e., $Y =S_i$ for some $i$;
this also means that $\# \lessdot a$;
\item Otherwise, let $b$ be the rightmost such leaf, and let $Y$ be $y$'s parent.
Notice that, $G$ being an operator grammar, $Y$ is the nearest common ancestor of $a$ and $b$.
%
%
Then there are two possibilities:
\begin{itemize}
\item Case 1: $X = Y$, i.e. $b \doteq a$;
\item Case 2: $X \neq Y$, and in this case $b$ has lower depth than $a$, so $b \lessdot a$.
\end{itemize}
\end{itemize}
In all cases, node $Z$ may be missing, or there may be other leaves between $b$ and $a$ (namely, $Z$'s descendants);
let $\hat Z = \perp$ if $Z$ is missing, $\hat Z = Z$ otherwise.
Then, for each such triple $(a,X,Y)$, define the \emph{$(a,X,Y)$-push transition}:
\[
\delta_{push}( (Y, \hat Z), a ) \ni 
	\left\{
	\begin{array}{ll}
		(X, X) & \text{if $a$ is the rightmost child of $X$,}\\
		(X, \perp) & \text{otherwise}.
	\end{array}
	\right.
\]
Hence, a push transition essentially determines the parent of the symbol under the input head 
(actually, a ``candidate'' parent, since the automaton is non-deterministic).



\begin{figure}[h]
\begin{center}

\begin{tabular}{ccc}
\begin{tikzpicture}[scale=0.5]
\node {$Y = S_i$}
child {
	node {$X$}
        edge from parent [path]
	child {
		node {$Z$}
		edge from parent [normal]
		child{
			node {$\dots$}
			edge from parent [normal]
		}
	}
	child {
		node {$a$}
		edge from parent [normal]
	}
	child {
		node {$\dots$}
		edge from parent [normal]
	}
}
child {node {$\dots$}};
\end{tikzpicture}

\qquad
&
\qquad

\begin{tikzpicture}[scale=0.5]
\node {}
child {
	node {$X = Y$} 
	child {
		node {$\dots$}
	}
	child {
		node {$b$}
	}
	child {
		node {$Z$}
		child{
			node {$\dots$}
		}
	}
	child {
		node {$a$}
	}
	child {
		node {$\dots$}
	}
};
\end{tikzpicture}

\qquad
&
\qquad

\begin{tikzpicture}[scale=0.5]
\node {}
child {
        node {$Y$}
	child {
		node {$\dots$}
	}
	child {
		node {$b$}
	}
	child {
		node {$W$}
		child {
			node {$X$}
		        edge from parent [path]
			child {
				node {$Z$}
		                edge from parent [normal]
				child{
					node {$\dots$}
			                edge from parent [normal]
				}
			}
			child {
				node {$a$}
			        edge from parent [normal]
			}
			child {
				node {$\dots$}
		                edge from parent [normal]
			}
		}
		child{
			node {$\dots$}
	                edge from parent [normal]
		}
	}
	child {
		node {$\dots$}
	}
};
\end{tikzpicture}
\\
Case 0 &Case 1 &Case 2\\
\end{tabular}

\end{center}
\caption{Derivation tree configurations for the push transition function (nodes labelled as $\dots$ could be missing).\label{fig:tree_push}}
\end{figure}
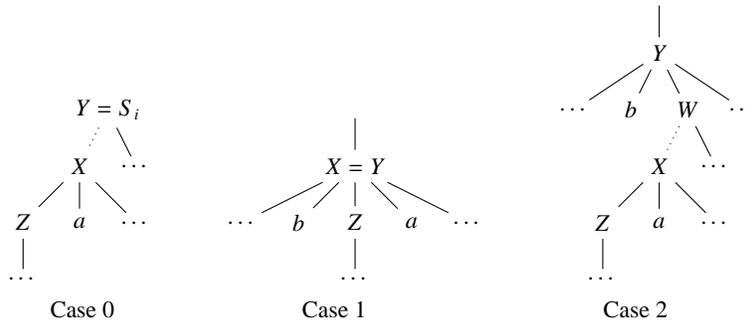


\noindent

A similar construction holds for the flush transition function $\delta_{flush}: Q \times Q \to 2^Q$.
For every derivation tree with internal node $X$, 
let $f$ and $\ell$ be the first and last child, respectively, of node $X$.
Notice that both $f$ and $\ell$ may be either internal nodes or leaves. 
Then there are two possibilities,
as depicted in Figure~\ref{fig:tree_flush}:
\begin{itemize}
\item Case 3: there is no leaf at the left of $X$, then let $Y$ be the topmost ancestor of $X$, i.e., $Y=S_i$ for some $i$;
\item Case 4: otherwise, let $b$ be the rightmost leaf at the left of $X$ and let $Y$ be $b$'s parent (again, notice that
$Y$ is the nearest common ancestor of $X$ and $b$, $G$ being an operator grammar).
\end{itemize}
Also, let $\ell_{/\!X}$ be $\ell$ if $\ell$ is an internal node, $X$ otherwise;
let $\tilde f$ be $f$ if $f$ is an internal node, $\perp$ otherwise.
Then, for each such pair $(X,Y)$ define the \emph{$(X,Y)$-flush transition}:
\[
\delta_{flush}( (X, \ell_{/\!X}), (Y, \tilde f)  ) \ni (Y, X).
\]
Hence, the state computed by a flush transition contains two pieces of information:
the first component determines the nearest ancestor of both $X$ and $b$ (or the axioms if $b$ does not exist),
while the second component determines the nonterminal corresponding to the r.h.s. just completed.


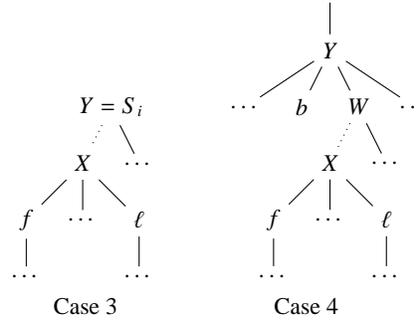
\begin{figure}[h]
\begin{center}
\tikzset{
	path/.style={dotted},
	every edge/.style={solid},
	normal/.style={solid},
}

\begin{tabular}{ccc}
\begin{tikzpicture}[scale=0.5]
\node {$Y=S_i$}
child {
	node {$X$}
        edge from parent [path]
	child {
		node {$f$}
		edge from parent [normal]
		child{
			node {$\dots$}
			edge from parent [normal]
		}
	}
	child {
		node {$\dots$}
		edge from parent [normal]
	}
	child {
		node {$\ell$}
		edge from parent [normal]
		child{
			node {$\dots$}
			edge from parent [normal]
		}
	}
}
child {node {$\dots$}};
\end{tikzpicture}

\qquad
&
\qquad
\begin{tikzpicture}[scale=0.5]
\node {}
child {
        node {$Y$}
	child {
		node {$\dots$}
	}
	child {
		node {$b$}
	}
	child {
		node {$W$}
child {
	node {$X$}
        edge from parent [path]
	child {
		node {$f$}
		edge from parent [normal]
		child{
			node {$\dots$}
			edge from parent [normal]
		}
	}
	child {
		node {$\dots$}
		edge from parent [normal]
	}
	child {
		node {$\ell$}
		edge from parent [normal]
		child{
			node {$\dots$}
			edge from parent [normal]
		}
	}
}
		child{
			node {$\dots$}
	                edge from parent [normal]
		}
	}
	child {
		node {$\dots$}
	}
};
\end{tikzpicture}
\\
Case 3 &Case 4\\
\end{tabular}

\end{center}
\caption{Derivation tree configurations for the flush transition function (all nodes marked as $\dots$ could be missing).\label{fig:tree_flush}}
\end{figure}


Finally, initial and final states are defined as follows.
\[
I = \{(S_i, \perp) \mid 1 \le i \le n(S)\}, 
\qquad 
F = \{(S_i, A_j) \mid S \to
A \in P, 1 \le i \le n(S), 1 \le j \le n(A)\}.
\]

Notice that the above construction is effective. 
All triples $(a,X,Y)$ involved by some push transition can be found starting from any rule $X \to \alpha$ with $\alpha$ containing $a$:
if $a$ is not the leftmost terminal of $\alpha$, then take the triple $(a, X, X)$, else apply backwards any rule
with r.h.s starting with $X$ and extend this process until all productions have been examined. Similarly for the flush transitions.

\begin{example}
Let $G$ be the grammar introduced in Example~\ref{ex:expr}.
Following the above construction, number the rules of the grammar in the order they appear in the definition of $G$
(for instance, $P_2(E)$ is $E \to T \times a$).
The transitions defined by the derivation tree of string $a \times a +a$, depicted in Figure~\ref{fig:string} (left),
are the following:
\[
\begin{array}{cc}
  \begin{array}{l}
  \delta_{push}( (S_1, \perp), a ) \ni (T_2, T_2) \\
  \delta_{push}( (S_1, T_2), \times ) \ni (E_2, \perp)\\
  \delta_{push}( (S_1, E_2), + ) \ni (E_1, \perp)\\
  \delta_{push}( (E_2, \perp), a ) \ni (E_2, E_2)\\
  \delta_{push}( (E_1, \perp), a ) \ni (T_2, T_2) \\
  \end{array}
& \qquad
  \begin{array}{l}
  \delta_{flush}( (T_2, T_2), (E_1, \perp)  ) \ni (E_1, T_2) \\
  \delta_{flush}( (T_2, T_2), (S_1, T_2)  ) \ni (S_1, T_2)\\
  \delta_{flush}( (E_2, E_2), (S_1, T_2)  ) \ni (S_1, E_2)\\
  \delta_{flush}( (E_1, T_2), (S_1, E_2)  ) \ni (S_1, E_1)\\
  \end{array}
\end{array}
\]
The first one is the $(a, T_2, S_1)$-push transition obtained by starting from  the left-most leaf (Case 0).
Case 0 occurs also for the second and third push transitions, obtained considering the leaves labeled by $\times$ and $+$, respectively.
The other push transitions represent instances of Cases 1 and 2, in this order.
As far as flush transitions are concerned, 
Case 4 occurs only in the first stated transition, with $X=T_2$, $b=+$ and $Y=E_1$,
whereas all other productions represent instances of Case 3.
Hence, on input $a \times a +a$, the automaton $\mathcal A$ obtained from $G$ may execute 
the computation represented in Figure~\ref{fig:string} (right).
\end{example}


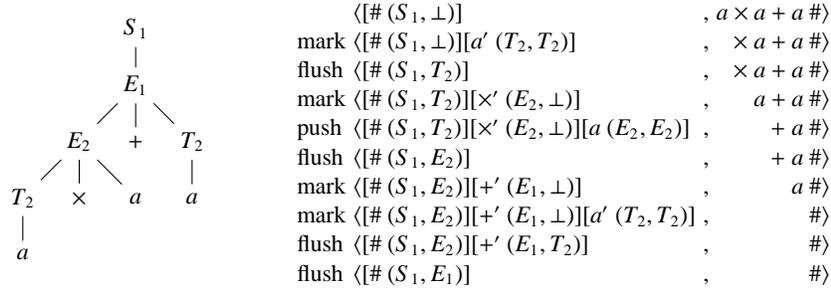
\begin{figure}[h]
\begin{center}
\begin{tabular}{m{.2\textwidth}m{.1\textwidth}m{.65\textwidth}}
\begin{tikzpicture}[scale=0.5]
\node {$S_1$}
child {
        node {$E_1$}
	child {
		node {$E_2$}
		child {
			node {$T_2$}
			child{
				node {$a$}
			}
		}
		child {
			node {$\times$}
		}
		child {
			node {$a$}
		}
	}
	child {
		node {$+$}
	}
	child {
		node {$T_2$}
		child {
			node {$a$}
		}
	}
};
\end{tikzpicture}
&
&
$
\begin{array}{llcr}
 & \langle \stack 0\# {(S_1, \perp)}     & , &    a \times a + a\ \# \rangle \\ 
\text{mark} & \langle \stack 0\# {(S_1, \perp)} \stack 1 a {(T_2, T_2)}     & , &      \times \ a + a\ \# \rangle \\
\text{flush} & \langle \stack 0\# {(S_1, T_2)}      & , &      \times \ a + a\ \# \rangle \\
\text{mark} & \langle \stack 0\# {(S_1, T_2)} \stack 1 \times {(E_2, \perp)}      & , &      a + a\ \# \rangle \\
\text{push} & \langle \stack 0\# {(S_1, T_2)} \stack 1 \times {(E_2, \perp)} \stack 0 a {(E_2,E_2)}     & , &      + \ a\ \# \rangle \\
\text{flush} & \langle \stack 0\# {(S_1,E_2)}    & , &     +\ a\ \# \rangle \\
\text{mark} & \langle \stack 0\# {(S_1,E_2)} \stack 1 + {(E_1,\perp)}    & , &       a\ \# \rangle \\
\text{mark} & \langle \stack 0\# {(S_1,E_2)} \stack 1 + {(E_1,\perp)} \stack 1 a {(T_2,T_2)}	& , &      \# \rangle \\
\text{flush} & \langle \stack 0\# {(S_1,E_2)} \stack 1 + {(E_1,T_2)}    & , &      \# \rangle \\
\text{flush} & \langle \stack 0\# {(S_1,E_1)}  & , &    \# \rangle \\
\end{array}
$
\end{tabular}
\end{center}
\caption{Derivation tree (left) and computation (right) for the string $a \times a + a$.\label{fig:string}}
\end{figure}


The equivalence between $G$ and the automaton described above is based on the following lemma,
whose proof is omitted because of space reasons. 
As usual we set $\Gamma = (\Sigma \cup \mrk \Sigma) \times Q = (\Sigma \cup \mrk \Sigma) \times (EN \times (EN \cup \{\perp\}))$
and we denote an element in $\Gamma$ as $\stack 0{a}{(X,Y)}$.
To avoid an excessively cumbersome notation, when describing the transitions
between configurations, we omit the extreme parts (i.e. the lower part of the
stack and a suffix of the input string) which are not affected by the
computation.

We define the \emph{depth of a computation} $C_1 \comp * C_2$ as the maximum number of marked symbols 
in one of the traversed configurations, minus the number of marked symbol on the stack in configuration $C_1$;
we define the \emph{depth of a derivation} $W \der * \alpha$ as the depth of the corresponding derivation tree.
When useful, we make the depth $h$ of a computation or a derivation explicit as in
$C_1 \comph h  C_2$ and $X \derh h \alpha$.
\begin{lemma}
\label{lemma:violi}
Let $Y,W$ be extended nonterminals of $G$, $v \in \Sigma^*$,
$a \lessdot v \gtrdot b$, and $\bar a \in \{a, \mrk a\}$. Then for all $h \geq 1$:
\[
\label{main}
	\config {\stack 0{\bar a}{(Y,\perp)}}{vb}
	\comph h	
	\config {\stack 0{\bar a}{(Y,W)}}{b}
	\quad
	\text{ iff }
	\quad
	\exists \alpha, \beta \text{ such that } Y \to  \alpha a W
\beta, \ W \derh h  v \text { in } G.
\]
\end{lemma}

\begin{proof}

The lemma is equivalent to the following two properties.
\begin{enumerate}[(i)]
\item
\label{down}
For every $Y,X$, $a \lessdot c \, \underline{\lessdot}\, x \gtrdot d$,
$\mathcal A$ admits the computation 
\[
\config {\stack 0{\bar a}{(Y,\perp)}}{c x d}
\comph k
\config {\stack 0{\bar a}{(Y,X)}}{d}
\]
if and only if there exist 
$W, \alpha, \beta, \gamma,\epsilon$ 
such that $Y \to  \alpha a W \beta,\  W \der * X \gamma,\  X \to c \epsilon,\
\epsilon \derh k x$.
\item%
\label{up}
For every $Y,X,Z$,
$a \lessdot d\, \underline{\lessdot}\, z \gtrdot e$,
$\mathcal A$ admits the computation 
\[
\config {\stack 0{\bar a}{(Y,X)}}{d z e}
\comph k
\config {\stack 0{\bar a}{(Y,Z)}}{d}
\]
if and only if there exist  
$W, \alpha, \beta,\mu, \lambda$ 
such that 
$Y \to  \alpha a W \beta,\ W \der *  Z \mu,\  Z \to X d \lambda, \lambda \derh k z$.
\end{enumerate}

\bigskip
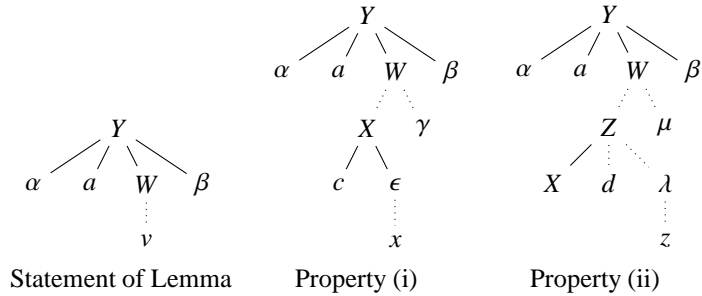
\begin{center}
\begin{tabular}{ccc}

\begin{tikzpicture}[scale=0.5]
\node {$Y$}
child {node {$\alpha$}}
child {node {$a$}}
child {node {$W$}
	child {
		node {$v$}
	        edge from parent [path]
	}
}
child {node {$\beta$}}
;
\end{tikzpicture}
\quad & \quad 
\begin{tikzpicture}[scale=0.5]
\node {$Y$}
child {node {$\alpha$}}
child {node {$a$}}
child {node {$W$}
	child {node {$X$}
		edge from parent [path]
		child {
			node {$c$}
			edge from parent [normal]
		}
		child { node {$\epsilon$} 
			edge from parent [normal]
			child { 
				node {$x$}
				edge from parent [path]
			}
		}
	}
	child {node {$\gamma$}
		edge from parent [path]
	}
}
child {node {$\beta$}
}
;
\end{tikzpicture}
\quad & \quad 
\begin{tikzpicture}[scale=0.5]
\node {$Y$}
child {node {$\alpha$}}
child {node {$a$}}
child {node {$W$}
	child {node {$Z$}
		edge from parent [path]
		child {
			node {$X$}
			edge from parent [normal]
		}
		child {
			node{$d$}
			edge from parent [path]
		}
		child {
			node {$\lambda$}
			edge from parent [path]
			child {
				node{$z$}
				edge from parent [path]
			}
		}
	}
	child { node {$\mu$} 
		edge from parent [path]
	}
}
child {node {$\beta$}}
;
\end{tikzpicture}
\\
Statement of Lemma
&
Property~\eqref{down}
&
Property~\eqref{up}
\end{tabular}

\end{center}
\bigskip

\noindent
Notice that in~\eqref{down} $W$ and $X$ may coincide (i.e., $\gamma$ may be empty), 
and in~\eqref{up} $W$ and $Z$ may coincide (i.e., $\mu$ may be empty).
For $h = 1$, the lemma is given by property~\eqref{down} with $W=X$ and $k = 0$
(for $cx = v$, $d=b$);
for $h > 1$ we have $v= cx d_1 z_1 d_2 \dots d_n z_n$ for some $c \, \underline\lessdot \, x \gtrdot d_1$, 
$d_i \, \underline\lessdot z_i \gtrdot d_{i+1}$ (with $x$, $z_i$ possibly empty).
Then, applying first property~\eqref{down} and then, repeatedly, property~\eqref{up}, one gets the lemma.

We prove property~\eqref{down} reasoning by induction on $k$.
First let  $k = 0$; in this case $\epsilon = x$, i.e. $X \to cx$.
Hence, if $x = c_1 \dots c_n$,  during the computation defined in~\eqref{down}, $\mathcal A$ has to execute the following series of moves:
a marked $(c,X_0,Y)$-push transition (case 2 without $Z$), then a sequence of $(c_i,X_0,X_0)$-push transitions (case 1 without $Z$), and 
finally a $(X_0,Y)$-flush transition, for a suitable $X_0$:
\[
\config {\eta 
	\stack 0{\bar a}{(Y,\perp)} 
	\stack 0{\mrk c}{(X_0,\perp)} 
	\stack 0{c_1}{(X_0,\perp)} 
	\dots
	\stack 0{c_n}{(X_0,X_0)}}{d}
\vdash
\config {\eta 
	\stack 0{\bar a}{(Y,X_0)} }{d}.
\]
To end in the right configuration, we necessarily have $X_0 = X$. Moreover, 
by the definition of transitions in $\mathcal A$, $X$ must satisfy exactly the relations defined in~\eqref{down}.
Vice versa, if the grammar admits the derivation defined in~\eqref{down}, then obviously the automaton $\mathcal A$ admits the previous moves.

One can prove similarly property~\eqref{up} for $k=0$: in this case, both the marked $(d,Z,Y)$-push transition and the 
$(Z,Y)$-flush transition involve the extended nonterminal $X$
(i.e., the second component of the state on the top of the stack).

Now, assuming that properties~\eqref{down} and~\eqref{up} hold for depths lower than $k$, we prove them for $k$.
First consider~\eqref{down} and let
$x = u_0 c_1 u_1 c_2 \dots c_m u_m$ with $c \lessdot u_0$, $u_{i-1} \gtrdot c_i \lessdot u_i$
(with any $u_i$ possibly empty), and $c_i \doteq c_{i+1}$.
By the definition of the transition function, 
$\mathcal A$ admits the computation in~\eqref{down} if and only if there exist
$W, \alpha, \beta, \gamma,\epsilon$ as in~\eqref{down} and moreover
there exist $U_0, \cdots U_m$ such that 
$\epsilon = U_0 c_1 U_1 \dots c_m U_m$ and $U_i \derh {k_i} u_i$ with $k_i < k$
($U_i$ is missing iff $u_i$ is empty).
Hence one can apply the inductive hypothesis and get the result.

One can prove similarly property~\eqref{up} for $k$ greater than 0: again, in this case, both the marked $(d,Z,Y)$-push transition and the 
$(Z,Y)$-flush transition involve the extended nonterminal $X$.
\qed
\end{proof}

\noindent From the lemma the theorem easily follows by using a special case
$S \to A$ (with implicit $\#$ as $a$ and $b$).

\subsection{From Floyd automata to Floyd grammars}

Given a Floyd automaton $\mathcal A = \langle \Sigma,M, Q, I, F, \delta \rangle $, 
we show how to build an equivalent Floyd grammar $G$ having operator precedence matrix M.
In order to keep the construction as easy as possible, w.l.o.g we assume that $M$ is $\dot=$-acyclic.
Remind that, as discussed in Section~\ref{sec:prelim}, this hypothesis could be replaced by weaker ones.

We need some notation and definitions.
First of all, we shall represent 
a push transition with a simple arrow $\rightarrow$,
a flush transition with a double arrow $\Rightarrow$,
and a path defined by a sequence of transitions with a wavy arrow $\leadsto$.

We define \emph{chains} in $\mathcal A$ recursively.
A \emph{simple chain} is a word $a_0 a_1 a_2 \dots a_n a_{n+1}$,
written as
$
\chain {a_0} {a_1 a_2 \dots a_n} {a_{n+1}},
$
such that:
$a_0, a_{n+1} \in \Sigma \cup \{\#\}$,
$a_i \in \Sigma$ for every $i = 1,2, \dots n$, 
$M_{a_0,a_{n+1}} \neq \emptyset$,
and $a_0 \lessdot a_1 \doteq a_2 \dots a_{n-1} \doteq a_n \gtrdot a_{n+1}$.
%
%
A \emph{composed chain}
in $\mathcal A$ is a word 
$a_0 x_0 a_1 x_1 a_2  \dots a_n x_n a_{n+1}$, 
where
$\chain {a_0}{a_1 a_2 \dots a_n}{a_{n+1}}$ is a simple chain, and
$x_i \in \Sigma^*$ is the empty word 
or is such that $\chain {a_i} {x_i} {a_{i+1}}$ is a chain (simple or composed),
for every $i = 0,1, \dots, n-1$. 
Such a composed chain will be written as
$\chain {a_0} {x_0 a_1 x_1 a_2 \dots a_n x_n} {a_{n+1}}$.

We call a \emph{support} for the simple chain
$\chain {a_0} {a_1 a_2 \dots a_n} {a_{n+1}}$
any path in $\mathcal A$ of the form
\begin{equation}
\label{eq:simplechain}
q_0
\va{a_1}{q_1}
\va{}{}
\dots
\va{}q_{n-1}
\va{a_{n}}{q_n}
\flush{q_0} {q_{n+1}}
\end{equation}
Notice that the label of the last (and only) flush is exactly $q_0$, i.e. the first state of the path; this flush is executed because of relation $a_n
\gtrdot a_{n+1}$.
We call a \emph{support for the composed chain} 
$\chain {a_0} {x_0 a_1 x_1 a_2 \dots a_n x_n} {a_{n+1}}$
any path in $\mathcal A$ of the form
\begin{equation}
\label{eq:compchain}
q_0
\ourpath{x_0}{q'_0}
\va{a_1}{q_1}
\ourpath{x_1}{q'_1}
\va{a_2}{}
\dots
\va{a_n} {q_n}
\ourpath{x_n}{q'_n}
\flush{q'_0}{q_{n+1}}
\end{equation}
where, for every $i = 0, 1, \dots, n$: 

\begin{itemize}
\item if $x_i \neq \epsilon$, then $q_i \ourpath{x_i}{q'_i} $ 
is a support for the chain $\chain {a_i} {x_i} {a_{i+1}}$, i.e.,
it can be decomposed as $q_i\ourpath{x_i}{q''_i} \flush{q_i}{q'_i}$.

\item if $x_i = \epsilon$, then $q'_i = q_i$.
\end{itemize}
Notice that the label of the last flush is exactly $q'_0$.
	
We are now able to define a Floyd grammar $G = \langle \Sigma, N, S, P \rangle$.
Nonterminals are the 4-tuples $(a, q, p, b) \in \Sigma \times Q \times Q \times \Sigma$, written as $\nont apqb$,
plus the axiom $S$.
Rules are built as follows:
\begin{itemize}
\item 
for every support of type~(\ref{eq:simplechain})
of a simple chain, add the rule 
\[
\nont {a_0}{q_0}{q_{n+1}}{a_{n+1}}\longrightarrow a_1 a_2 \dots a_n\ ;
\]
if also $a_0 = a_{n+1} = \#$, $q_0$ is initial, and $q_{n+1}$ is final, add the rule
$S \rightarrow \nont {\#}{q_0}{q_{n+1}}{\#}$;
\\
\item 
for every support of type~(\ref{eq:compchain}) of a composed chain,
add the rule 
\[
\nont {a_0}{q_0}{q_{n+1}}{a_{n+1}} \longrightarrow N_0 a_1 N_1 a_2 \dots a_n N_n \ ;
\]
where, for every $i = 0,1, \dots, n$,
$N_i = \nont {a_i}{q_i}{q'_i}{a_{i+1}}$ if $x_i \neq \epsilon$ and $N_i =
\epsilon$ otherwise.
\end{itemize}

Notice that the above construction is effective thanks to the hypothesis of
$\dot=$-acyclicity of the OPM. This implies that the length of the r.h.s. is
bounded (see Section \ref{sec:prelim}); on the other hand,
the cardinality of the nonterminal alphabet is finite.   Hence there is only a finite
number of possible productions for $G$ and only a limited number of chains to be
considered.


\section{$\omega$-languages}
\label{sec:omega}

Having an operational model that defines Floyd Languages, it is now
straightforward to introduce extensions to $\omega$-languages.

For instance, the classical B\"uchi condition of acceptance can be easily adapted to
FAs. Consider an infinite word $x \in \Sigma^\omega$, and an {\em infinite computation} 
of the automaton $\mathcal A_{M} = \langle \Sigma, M, Q, I, F, \delta \rangle $
on $x$, i.e. an $\omega$-sequence of configurations 
$\mathcal S = \config {\beta_0} {x_0} \config {\beta_1} {x_1} \ldots$, such that  
$\config {\beta_0} {x_0} = \config {\stack 0\#{q_I}} {x}$, $q_I \in I$ and 
$\config {\beta_i} {x_i} \vdash \config {\beta_{i+1}} {x_{i+1}}$. 
We say that $x \in L(\mathcal A)$ if and only if there exists 
$q_F \in F$ such that 
configurations with stack $\stack 0\#{q_F}$ occur infinitely often in $\mathcal
S$.

Quite naturally, $\omega$-VPLs are a proper subset of this class of languages,
as it is shown by the following example.

\begin{example}
We define here the stack management of a simple programming language that is able to
handle nested exceptions. For simplicity, there are only two procedures,
called $a$ and $b$. Calls and returns are denoted by $call_a$, $call_b$,
$ret_a$, $ret_b$, respectively.  
During execution, it is possible to install an exception handler $hnd$. The
last signal that we use is $rst$, that is issued when an exception occur,
or after a correct execution to uninstall the handler. With a $rst$ the stack is
``flushed'', restoring the state right before the last $hnd$. The automaton is
presented in Figure \ref{ex:omega} (notice that it is an extension of the
automaton in Figure~\ref{ex:primo}).
It is easy to modify this example to model the case of {\em unnested}
exceptions, to fit with other application contexts.
\end{example}


\begin{figure}
\begin{center}
\begin{tabular}{m{.4\textwidth}m{.45\textwidth}}
$
\begin{array}{c|cccccc}
      & call_a & ret_a & call_b & ret_b & hnd & rst  \\
\hline
call_a & \lessdot & \dot= & \lessdot &  & \lessdot & \gtrdot \\
ret_a  & \gtrdot & \gtrdot & \gtrdot & \gtrdot & \gtrdot & \gtrdot \\
call_b & \lessdot & & \lessdot & \dot=  & \lessdot & \gtrdot \\
ret_b  & \gtrdot & \gtrdot & \gtrdot & \gtrdot & \gtrdot & \gtrdot \\
hnd    & \lessdot & \lessdot & \lessdot & \lessdot & & \dot= \\
rst    & \gtrdot & \gtrdot & \gtrdot & \gtrdot \\
\#       & \lessdot & & \lessdot & & \lessdot \\
\end{array}
$
&
\begin{tikzpicture}[every edge/.style={draw,solid}, node distance=4cm, auto, 
                    every state/.style={draw=black!100,scale=0.5}, >=stealth]

\node[initial by arrow, initial text=,state,accepting] (S) {{\huge $q_0$}};
\node[state] (E) [right of=S, xshift=3cm] {{\huge $q_1$}};

\path[->]
(S) edge [bend left]  node {$call_a, call_b, hnd$} (E)
(E) edge [loop right, double] node {$q_1$} (E)
(E) edge [loop below] node {$call_a, ret_a, call_b, ret_b, hnd, rst$} (E)
(E) edge [below, double]  node {$q_0$} (S) ;
\end{tikzpicture}
\end{tabular} 
\caption{Precedence matrix and automaton for an $\omega$-language. There is no
column indexed by \# since words are infinite.}\label{ex:omega}
\end{center}
\end{figure}
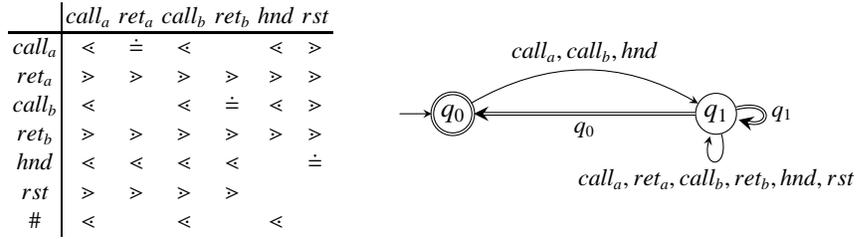


\section{Conclusions and further research}
\label{sec:conclusions}

Recently, we advocated that operator precedence grammars and languages, here renamed after their inventor Robert Floyd, deserve renewed attention in the realm of formal languages. The main reasons to support our claim are:

\begin{itemize}
\item
The fact that this family of languages properly includes visibly pushdown languages 
\cite{jacm/AlurM09}, a new family that has been proposed with the main
motivation of extending powerful model checking techniques beyond the limits of
finite state machines.
\item
The fact that it enjoys all closure properties with respect to the main
algebraic operations that are exhibited by regular languages and VPLs.
\item
The fact that, unlike other deterministic languages -either strictly more powerful than them, or incomparable with them- such as LR, LL, and simple precedence ones, FLs can be parsed without applying a strictly left-to-right order; this feature becomes particularly relevant in these days since it allows to exploit much better the gains in efficiency offered by massive parallelism.
\end{itemize}

\noindent
In this paper we filled a rather surprising ``hole'' in the theory of these languages, namely the lack of an appropriate family of automata that perfectly matches the generative power of their grammars. We defined FAs with such a goal in mind and we proved their equivalence with FGs. Both facts turned out to be non-trivial jobs and showed further interesting peculiarities of this pioneering family of deterministic languages.
	A first ``byproduct'' of the new automata family is the extension of FLs
to $\omega$-languages, i.e., languages consisting of infinite strings, a more
and more important aspect of formal language theory needed to deal with
never ending computations. In this case too FL $\omega$-languages proved to augment the descriptive capabilities of the original VPLs.

As a first step towards applicability of the results presented in this paper,
and also to validate our approach with several practical examples, we
implemented a simple prototypical tool, called {\em Flup}. Flup contains an
interpreter for non-deterministic Floyd Automata, and a Floyd Grammar to
Automata translator, that directly applies the construction presented in
Section~\ref{sec:gr2aut}.
 All the examples presented in the paper
were tried on, or generated by the tool.%
\footnote{The prototype is freely available at {\em
http://home.dei.polimi.it/pradella}.}

We are confident that suitable future research will further strengthen the
importance of, and motivation for, re-inserting FLs in the main stream of formal
language literature. In particular it would be interesting to complete the
parallel analysis and comparison with VPLs by investigating a characterization
in terms of suitable logic formulas \cite{jacm/AlurM09}; by this way motivation for, and application of, strong model checking techniques would be further enhanced.

\paragraph{Acknowledgement.}
We thank Federica Panella for her comments and suggestions, especially with respect to the construction of Theorem\~ref{th:nondet} and $\omega$-languages.

\bibliographystyle{plain}
\bibliography{VPDbib}

\newpage
\section*{Appendix: proof of Theorem~\ref{th:nondet}}

\paragraph{Notation.}
We use $J, \bar J, J', J_i \dots$ to denote states in $\tilde Q$ and $K, \bar K, K', K_i, \dots$ to denote set of pairs in $Q \times (Q \cup \{\bot\})$.
We use arrows $\va{}{}$ and $\flush{}{}$ to denote push and flush transitions, respectively, both in $\mathcal A$ and in $\tilde{\mathcal A}$.

\paragraph{Remarks}
\begin{enumerate}[i)]
\item
\label{rem:flush}
By the definition of $\delta_{\text{flush}}$, if 
$\pair{b}{\bar K} \flush{\pair{a}{K}}{\pair{a}{K'}}$ in $\tilde{\mathcal A}$
and $\stpair{q'}{p} \in K'$,
then
there exists a pair $\stpair r q \in \bar K$ such that $\stpair q p \in K$ and 
$r \flush{q}{q'}$ in $\mathcal A$.
\item
\label{rem:push1}
By the definition of $\delta_{\text{push}}$, if 
$\pair{b}{\bar K} \va{a}{\pair{a}{K}}$ in $\tilde{\mathcal A}$,
$\stpair rq \in K$,
and $b \doteq a$, then
there exists a state $\bar q \in Q$ such that $\bar q \va a r$ in $\mathcal A$
and $\stpair{\bar q}{q} \in \bar K$.
\item
\label{rem:push2}
By the definition of $\delta_{\text{push}}$, if 
$\pair{b}{\bar K} \va{a}{\pair{a}{K}}$ in $\tilde{\mathcal A}$,
$\stpair{\bar q}{q} \in K$, $\stpair qp \in \bar K$, 
and $b \lessdot a$, then
$q \va{a}{\bar q}$ in $\mathcal A$.
\end{enumerate}

\begin{lemma}
\label{lem:rem}
Let $\mathcal C = \chain{a}{y}{b}$ be a chain and let $\pair{a}{K} \ourpath{y} \pair{a'}{K'} $
be a support for $\mathcal C$ in $\mathcal{\tilde{A}}$. Then $a' = a$ and the support has the form
{\small
\begin{equation}
\label{eq:support}
\pair a K
\ourpath{x_0}{\pair {a}{\bar K_0}}
\va{a_1}{\pair {a_1}{K_1}}
\ourpath{x_1}{\pair {a_1}{\bar K_1}}
\va{a_2}{}
\dots
\va{a_n} {\pair{a_n}{K_n}}
\ourpath{x_n}{\pair {a_n}{\bar K_n}}
\flush{\pair {a}{\bar K_0}}{\pair{a}{K'}}
\end{equation}
}where $y = x_0 a_1 x_1 a_2 \dots a_n x_n $ and $\chain a {a_1 a_2 \dots a_n} b$ is a simple chain.
Any word $x_i$ may be empty and in this case $\bar K_i = K_i$.
\end{lemma}

\begin{proof}
We argue by induction on the number of flush transitions in the support.
If there is only one flush transition, then the chain is simple, i.e.
$y = a_1 a_2 \dots a_n$ with $a \lessdot a_1 \doteq a_2 \doteq \dots \doteq a_n \gtrdot b$, and 
by the definition of $\tilde\delta_{\text{push}}$, the support can be rewritten as\begin{equation}
\label{eq:supportsimple}
\pair{a}{K}
\va{a_1}{\pair{a_1}{K_1}}
\va{a_2}{}
\dots
\va{a_{n-1}}\pair{a_{n-1}}{K_{n-1}}
\va{a_{n}}\pair{a_n}{K_n}
\flush{\pair{a}{K}} \pair{a'}{K'}
\end{equation}
By the definition of $\tilde\delta_{\text{flush}}$, we get $a'=a$.

Now assume that the statement holds for supports with $k$ flush transitions at most.
Let
$y = x_0 a_1 x_1 a_2 \dots a_n x_n $, where $\chain {a}{a_1 a_2 \dots a_n}{b}$ is a simple chain,
and consider the support
\[
\pair a K
\ourpath{x_0}{{\bar J}_0}
\va{a_1}{\pair {a_1}{K_1}}
\ourpath{x_1}{{\bar J}_1}
\va{a_2}{}
\dots
\va{a_n} {\pair{a_n}{K_n}}
\ourpath{x_n}{{\bar J}_n}
\flush{{\bar J}_0}{\pair{a'}{K'}}
\]
where,
for every $i =  0,1,2,\dots n$, the support labeled by $x_i$ contains $k$ flush transitions at most.
The inductive hypothesis implies that $\bar J_i$ has the form $\pair{a_i}{\bar K_i}$ for some $\bar K_i$ (where $a_0 =a$).
In particular the state $\bar J_0$
has the form $\pair {a}{\bar K_0}$ hence, by the definition of $\tilde\delta_{\text{flush}}$, we have $a'=a$.
\qed
\end{proof}

\begin{lemma}
\label{lem:chainnondet}
Let $\mathcal C = \chain{a}{y}{b}$ be a chain
and $q \ourpath{y}{q'}$ be a support for $\mathcal C$ in $\mathcal{A}$.
Then, for every $p \in Q$ and 
$K \subseteq Q\times(Q \cup \{\bot\})$, if $K \ni \stpair q p$, there exists a support 
\[\pair{a}{K} \ourpath{y} \pair{a}{K'} 
\]
for $\mathcal C$ in $\mathcal{\tilde{A}}$ with $K' \ni \stpair{q'}{p}$.
\end{lemma}

\begin{proof}
We argue by induction on the number of flush transitions contained in the support $q \ourpath{y}{q'}$.
If there is only one flush transition, then 
$y = a_1 a_2 \dots a_n$ with $a \lessdot a_1 \doteq a_2 \doteq \dots \doteq a_n \gtrdot b$ and the support can be rewritten as
\[
q = q_0
\va{a_1}{q_1}
\va{a_2}{}
\dots
\va{a_{n-1}}q_{n-1}
\va{a_{n}}{q_n}
\flush{q_0} {q'}
\]
Set $K_0 = K$, $a_0 = a$, and 
\begin{eqnarray*}
\pair{a_i}{K_i} &=& \tilde\delta_{\text{push}} ( \pair{a_{i-1}}{K_{i-1}} , a_i )  , \text{ for every } i=1, 2, \dots, n\\
\pair{a}{K'} &=& \tilde\delta_{\text{flush}} (\pair{a_n}{K_n}, \pair a K)
\end{eqnarray*}
Then 
\[
\pair{a}{K}
\va{a_1}{\pair{a_1}{K_1}}
\va{a_2}{}
\dots
\va{a_{n-1}}\pair{a_{n-1}}{K_{n-1}}
\va{a_{n}}\pair{a_n}{K_n}
\flush{\pair{a}{K}} \pair{a}{K'}
\]
is a support for $\mathcal C$ in $\mathcal{\tilde{A}}$.
Moreover, since $K \ni \stpair q p$, by the definition of $\tilde\delta$ we have:
\[
\begin{array}{ll}
K_1 \ni \stpair{q_1}{q} 
	&\text{since } a \lessdot a_1 \text{ and } \delta_{\text{push}}(q,a_1) \ni q_1\\
K_i \ni \stpair{q_i}{q} 
	&\text{since } a_{i-1} \doteq a_i \text{ and } \delta_{\text{push}}(q_{i-1},a_i) \ni q_i\\
K' \ni \stpair{q'}{p}   
	&\text{since  } a_n \gtrdot b \text{ and } \delta_{\text{flush}}(q_n,q) \ni q'
\end{array}
\]

Now assume that the statement holds for supports with $k$ flush transitions at most.
Let
$y = x_0 a_1 x_1 a_2 \dots a_n x_n $, where $\chain {a}{a_1 a_2 \dots a_n}{b}$ is a simple chain,
and consider the support
\[
q = q_0
\ourpath{x_0}{{\bar q}_0}
\va{a_1}{q_1}
\ourpath{x_1}{{\bar q}_1}
\va{a_2}{}
\dots
\va{a_n} {q_n}
\ourpath{x_n}{{\bar q}_n}
\flush{{\bar q}_0}{q'}
\]
where
${\bar q}_i = q_i$ whenever $x_i$ is the empty word and,
for every $i =  0,1,2,\dots n$, the support labeled by $x_i$ contains $k$ flush transitions at most.

Set 
$J_0 = \pair{a}{K}$ and 
\begin{eqnarray*}
\bar J_i &=& \tilde\delta (J_i,x_i) \text{ for every } i= 0, 1, \dots, n\\
J_i &=& \tilde\delta_{\text{push}} (\bar J_{i-1},a_i) \text{ for every } i=1, 2, \dots, n\\
J' &=& \tilde\delta_{\text{flush}} (\bar J_n,	\bar J_0).
\end{eqnarray*}
Then
\begin{equation}
\label{eq:supportJ}
\pair a K
\ourpath{x_0}{{\bar J}_0}
\va{a_1}{J_1}
\ourpath{x_1}{{\bar J}_1}
\va{a_2}{}
\dots
\va{a_n} {J_n}
\ourpath{x_n}{{\bar J}_n}
\flush{{\bar J}_0}{{J'}}
\end{equation}
is a support of $\mathcal C$ in $\mathcal{\tilde A}$.
By Lemma~\ref{lem:rem}, there exist $K_i, \bar K_i, K'$ such that 
$\bar J_i = \pair {a_i}{\bar K_i}$, $J_i = \pair {a_i}{K_i}$,  and $J' = \pair {a}{K'}$, where $a_0 = a$,
i.e., the support is \eqref{eq:support}. 

Moreover, since $K \ni \stpair q p$, by the definition of $\tilde\delta$ we have:
\[
\begin{array}{ll}
\bar K_0 \ni  \stpair {\bar q_0} p &
	\text{by inductive hypothesis on the support } q = q_0 \ourpath{x_0}{{\bar q}_0}\\
K_1 \ni \stpair {q_1} {\bar q_0} & 
	\text{since } a_0 \lessdot a_1 \text{ and }\delta_{\text{push}}({\bar q}_0,a_1) \ni q_1\\
\bar K_1 \ni \stpair {\bar q_1}{\bar q_0} &
	\text{by inductive hypothesis on the support }q_1 \ourpath{x_1}{{\bar q}_1}\\
K_i \ni \stpair{q_i}{\bar q_0}&
	\text{since } a_{i-1} \doteq a_i \text{ and } \delta_{\text{push}}({\bar q}_{i-1},a_i) \ni q_i\\
\bar K_i \ni \stpair{\bar q_i}{\bar q_0} &
	\text{by inductive hypothesis on the support } q_i \ourpath{x_i}{{\bar q}_i}\\
K' \ni \stpair{q'}{p} &
	\text{since } \delta_{\text{flush}}(q_n,{\bar q}_0) \ni q'
\end{array}
\]
\qed
\end{proof}

\begin{lemma}
\label{lem:chaindet}
Let $\mathcal C = \chain{a}{y}{b}$ be a chain 
and  $\pair a K \ourpath{y}\pair{a}{K'}$ be a support for $\mathcal C$ in $\mathcal{\tilde{A}}$. Then, for every $p,q' \in Q$,
if $K' \ni \stpair{q'}{p}$ there exists a support 
$q \ourpath{y}{q'}$ for $\mathcal C$ in $\mathcal{A}$ 
with $\stpair q p \in K$.
\end{lemma}

\begin{proof}
We argue by induction on the number of flush transitions contained in the support $\pair a K \ourpath{y}\pair{a}{K'}$.
If there is only one flush transition, then 
$y = a_1 a_2 \dots a_n$ with $a_0 \lessdot a_1 \doteq a_2 \doteq\dots\doteq a_n \gtrdot a_{n+1}$ and the support can be rewritten as
in \eqref{eq:supportsimple}.
Let $K' \ni \stpair{q'}{p}$;
then, by remark~\eqref{rem:flush}
there exists a pair $\stpair{q_n}{q} \in K_n$ such that
$\stpair qp \in K$ and $q_n \flush{q}{q'}$ in $\tilde{\mathcal A}$. 
Moreover, $\stpair{q_n}{q} \in K_n$, 
$\pair{a_{n-1}}{K_{n-1}} \va{a_n}{\pair{a_n}{K_n}}$ and $a_{n-1} \doteq a_n$ imply by remark~\eqref{rem:push1}
the existence of a state $q_{n-1} \in Q$ such that 
$\stpair{q_{n-1}}{q} \in K_{n-1}$ and
$q_{n-1}\va{a_n}{q_n}$.
Similarly one can verify that for every $i = n-2, \dots 1$ there exists $q_i \in Q$ such that
$\stpair{q_i}{q} \in K_i$ and
$q_i \va{a_{i+1}}{q_{i+1}}$.
In particular, $\pair aK \va{a_1}{\pair{a_1}{K_1}}$,
$\stpair{q_1}{q} \in K_1$,
$\stpair qp \in K$,
and $a \lessdot a_1$ imply by Remark~\eqref{rem:push2} that $q \va{a_1}{q_1}$ in $\mathcal A$.
Thus, we built backward a path
\[
q
\va{a_1}{q_1}
\va{a_2}{q_2}
\va{a_3}{}
\dots
\va{a_n}{q_n}
\flush{q} {q'}
\]
with $\stpair q p \in K$,
and this concludes the proof of induction basis.

Now assume that the statement holds for supports with $k$ flush transitions at most.
Let
$y = x_0 a_1 x_1 a_2 \dots a_n x_n $, where $\chain {a}{a_1 a_2 \dots a_n}{b}$ is a simple chain,
and consider a support of the form
\[
\pair a K
\ourpath{x_0}{{\bar J}_0}
\va{a_1}{J_1}
\ourpath{x_1}{{\bar J}_1}
\va{a_2}{}
\dots
\va{a_n} {J_n}
\ourpath{x_n}{{\bar J}_n}
\flush{{\bar J}_0}{{\pair {a} {K'}}}
\]
where
${\bar q}_i = q_i$ whenever $x_i$ is the empty word and,
for every $i =  0,1,2,\dots n$, the support labeled by $x_i$ contains $k$ flush transitions at most.
Then by Lemma~\ref{lem:rem}
the support can be rewritten as in~\eqref{eq:support}.

Let $\stpair{q'}{p} \in K'$. 
Since $\pair{a_n}{\bar K_n} \flush{\pair{a}{\bar K_0}}{\pair{a}{K'}}$, by Remark~\eqref{rem:flush} there exists
a pair $\stpair{\bar q_n}{\bar q_0} \in \bar K_n$ with $\stpair{\bar q_0}{p} \in \bar K_0$ and
$\bar q_n \flush{\bar q_0}{q'}$ in $\tilde{\mathcal A}$.
By the inductive hypothesis, since $\stpair{\bar q_n}{\bar q_0} \in \bar K_n$ there exists a support $q_n \ourpath{x_n}{\bar q_n}$ with $\stpair{q_n}{\bar q_0}\in K_n$.

Similarly one can see that, for all $i = n-1, \dots 2,1$, there exist $\bar q_i$ and $q_i$  such that
\[
q_i \ourpath{x_i}{\bar q_i} \va{a_{i+1}}{q_{i+1}}
\]
with 
$
\stpair{\bar q_i}{\bar q_0} \in \bar K_i
$
by Remark~\eqref{rem:push1} (since $\pair{a_i}{\bar K_i} \va{a_{i+1}}{\pair{a_{i+1}}{K_{i+1}}}$ in $\tilde{\mathcal A}$, 
$\stpair {q_{i+1}}{\bar q_0} \in K_{i+1}$,
and $a_i \lessdot a_{i+1}$), and
$
\stpair{q_i}{\bar q_0} \in K_i
$
by the inductive hypothesis (since $\pair{a_i}{K_i} \ourpath{x_i}{\pair{a_i}{\bar K_i}}$ in $\tilde{\mathcal A}$ and 
$\stpair {\bar q_i}{\bar q_0} \in \bar K_i$).

In particular $q_1 \ourpath{x_1}{\bar q_1}$ with $\stpair{q_1}{\bar q_0} \in K_1$. Then, since also $\pair{a}{\bar K_0} \va{a_1}{\pair{a_1}{K_1}}$,
$\stpair{\bar q_0}{p} \in \bar K_0$, 
and $a \lessdot a_1$, 
by Remark~\eqref{rem:push2} we get $\bar q_0 \va{a_1}{q_1}$.
Finally, since $\stpair{\bar q_0}{p} \in \bar K_0$ and $\pair aK \ourpath{x_0}{\pair{a}{\bar K_0}}$, the inductive hypothesis implies the existence of a state $q\in Q$ such that
$q \ourpath{x_0}{\bar q_0}$ in $\tilde{\mathcal A}$ with $\stpair qp \in K$.
Hence we built a support
\[
q
\ourpath{x_0}{{\bar q}_0}
\va{a_1}{q_1}
\ourpath{x_1}{{\bar q}_1}
\va{a_2}{}
\dots
\va{a_{n-1}}{q_{n-1}}
\ourpath{x_{n-1}}{\bar q_{n-1}}
\va{a_n} {q_n}
\ourpath{x_n}{{\bar q}_n}
\flush{{\bar q}_0}{q'}
\]
with $\stpair qp \in K$ and this concludes the proof.
\qed
\end{proof}

We are now ready to prove Theorem~\ref{th:nondet}, i.e., we prove that there exists an accepting computation for $y$ in $\mathcal A$
if and only if there exists an accepting computation for $y$ in $\mathcal{\tilde A}$.

Let $c$ be an accepting computation for $y$ in $\mathcal A$.
Then for $K = I \times \{\bot\} \ni \stpair{q_0}{\bot}$ Lemma~\ref{lem:chaindet} implies
the existence of a support $\tilde I = \pair{\#}{K} \ourpath{y}{\pair{\#}{K'}}$ for $y$ in $\mathcal{\tilde A}$ 
with $K' \ni \stpair{q'}{\bot}$. $q' \in F$ implies $\pair{\#}{K'}\in \tilde F$, hence the support defines an accepting computation for $y$ in $\mathcal{\tilde A}$.

Vice versa, let $c$ be an accepting computation for $y$ in $\mathcal{\tilde A}$. Then  
$\chain{\#}{y}{\#}$ is a chain that admits a support $\tilde I \ourpath{y}{J'}$ in $\mathcal{\tilde A}$, with $J' \in \tilde F$.
This means that there exists $q' \in F$ such that $\langle \#, q', \bot \rangle \in J'$. Hence, by Lemma~\ref{lem:chainnondet}, there exists a support 
$q \ourpath{y}{q'}$ in $\mathcal A$ with $\langle \#, q, \bot \rangle \in \tilde I$, and this implies $q \in I$.
Thus the support $q \ourpath{y}{q'}$ defines an accepting computation for $y$ in $\mathcal A$.
\qed

\end{document}